\documentclass{jfm}

\usepackage{graphicx}
\usepackage{newtxtext}
\usepackage{newtxmath}
\usepackage{natbib}
\usepackage{hyperref}
\usepackage{mathtools}
\usepackage{commath}
\usepackage{todonotes}
\usepackage[justification=justified]{caption}

\hypersetup{
    colorlinks = true,
    urlcolor   = blue,
    citecolor  = black,
}
\newtheorem{lemma}{Lemma}
\newtheorem{theorem}{Theorem}

\newcommand{\RomanNumeralCaps}[1]
\linenumbers

\renewcommand{\aa}{\mathbf{a}}
\newcommand{\bd}{\partial}
\newcommand{\dd}{\mathbf{d}}
\newcommand{\DD}{\mathcal{D}}
\newcommand{\eeta}{\boldsymbol{\eta}}
\newcommand{\FF}{\mathbf{F}}

\newcommand{\JJ}{\mathbf{J}}
\newcommand{\nnu}{\boldsymbol{\nu}}
\newcommand{\ttau}{\boldsymbol{\tau}}
\newcommand{\ssigma}{\boldsymbol{\sigma}}
\newcommand{\rr}{\mathbf{r}}
\newcommand{\RR}{\mathbf{R}}
\renewcommand{\SS}{\mathbf{S}}
\newcommand{\xx}{\mathbf{x}}

\newcommand{\uu}{\mathbf{u}}
\renewcommand{\vv}{\mathbf{v}}
\newcommand{\yy}{\mathbf{y}}
\newcommand{\pderiv}[2]{\frac{\partial #1}{\partial #2}}
\newcommand{\jump}[1]{[\![ #1 ]\!]}

\title{Two-Dimensional Vesicle Hydrodynamics from Hydrophobic Attraction Potential}
\author{
Szu-Pei Fu\aff{1},
Bryan Quaife\aff{2},
Rolf Ryham\aff{1}, \and
Yuan-Nan Young\aff{3}
}
 \affiliation{
\aff{1}Department of Mathematics, \\Fordham University, Bronx, New York 10458, USA
\aff{2}Department of Scientific Computing, \\Florida State University, Tallahassee, Florida 32306, USA
\aff{3}Department of Mathematical Sciences, New Jersey Institute of Technology,\\ Newark, New Jersey 07102, USA
 }

\begin{document}

\maketitle

\begin{abstract}
  We develop a new model, to our knowledge,  for the many-body  hydrodynamics of amphiphilic 
  Janus particles suspended in a viscous background flow. 
  The Janus particles interact through a hydrophobic attraction potential 
  that leads to self-assembly into bilayer structures.   
  We adopt an efficient integral equation method for solving the screened
  Laplace equation for hydrophobic attraction and 
  for solving the mobility problem for hydrodynamic interactions. 
  The integral equation formulation accurately captures both interactions for near touched boundaries. 
  Under a linear shear flow,  we observe the tank-treading deformation
  in a two-dimensional vesicle made of Janus particles.
  The results yield measurements of inter-monolayer friction, membrane permeability, 
  and at large shear rates, membrane rupture. 
  The simulations studies include a vesicle in parabolic flow and vesicle-vesicle 
  interactions in shear and extensional flows. The hydrodynamics of the Janus 
  particles vesicle replicate the behaviour of an inextensible elastic vesicle membrane.
\end{abstract}

\begin{keywords}
Authors should not enter keywords on the manuscript, as these must be chosen by the author during the online submission process and will then be added during the typesetting process (see \href{https://www.cambridge.org/core/journals/journal-of-fluid-mechanics/information/list-of-keywords}{Keyword PDF} for the full list).  Other classifications will be added at the same time.
\end{keywords}

{\bf MSC Codes }  {\it(Optional)} Please enter your MSC Codes here

\section{\label{intro}Introduction}
Described by physicist Pierre-Gilles de Gennes as ``another animal in soft matter physics", the
Janus particle--often a spherical particle with a hydrophobic and a
hydrophilic hemisphere--exhibits complex aggregate, clustering, and
self-assembly into mesoscopic and macroscopic structures that are
relevant to a wide range of applications in biology and bioengineering (\cite{deGennes1991}).
Whether it is surface chemistry or polarity under an external field, the
dynamics of Janus particles in a viscous solvent is inevitably the combination of long-range
hydrodynamics interactions with both short- and intermediate-range 
particle-particle interactions. Such multi-scale nature of Janus particle dynamics underlies 
the richness of a Janus particle suspension, as \cite{deGennes1991} suggested by the example of 
a ``thin film of Janus grains" that can breathe due to the interstices between Janus particles.

\begin{figure}
\begin{center}
\includegraphics[width=\textwidth]{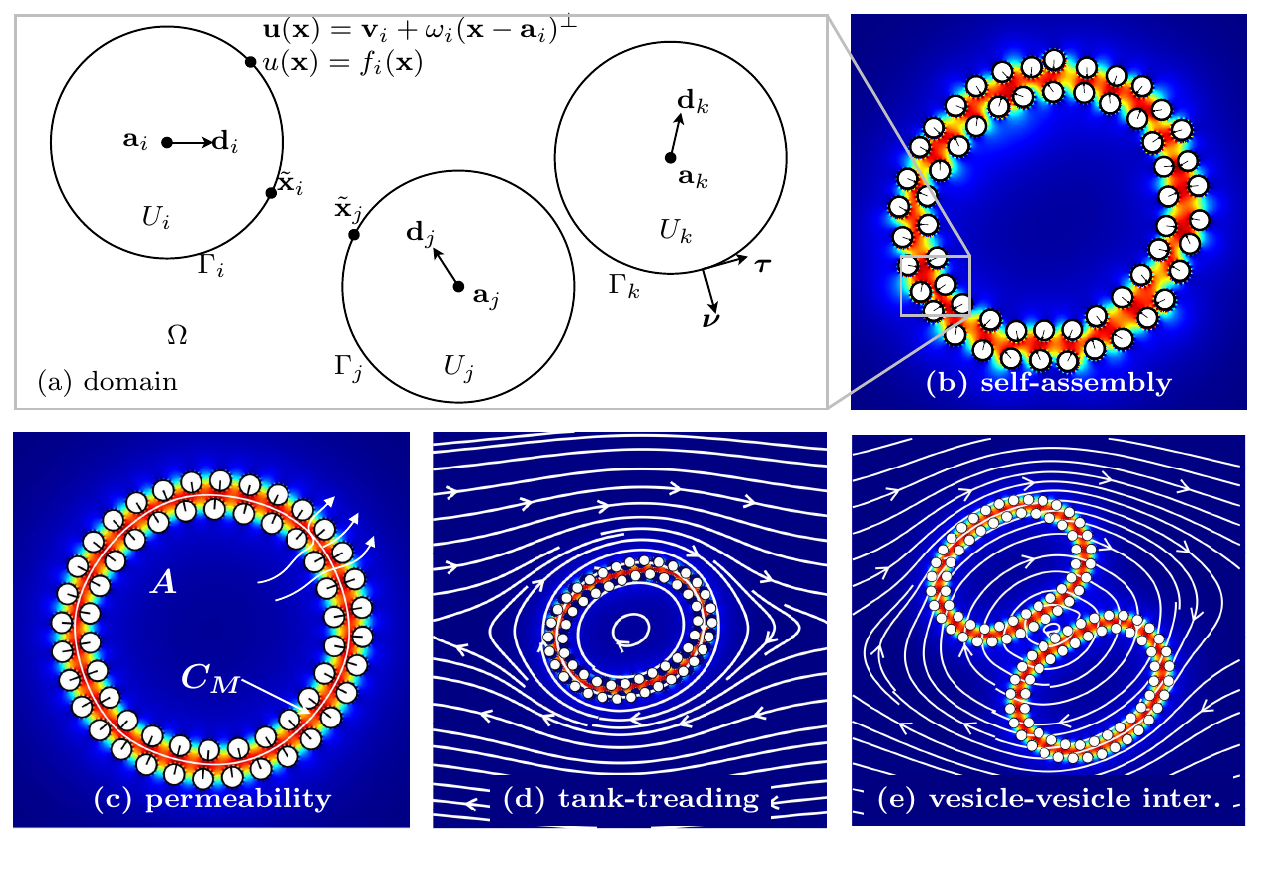}
\end{center}
\caption{Panel (a) illustrates the domain and boundary conditions for
  the mobility problem. The color map in panels (b)--(e) is the solution
  $u$ of \eqref{eq:SL}. Red is for $u = 1$ and blue is for $u = 0$.
  Particles self-assemble into vesicle bilayers (panel (b)) and
  eventually arrange along inner and outer leaflets. In panel (c), $C_M$
  is the midplane curve, $A$ is the enclosed area, and $L$ is the arc
  length. Initially stretched vesicles relax to equilibrium with fluid
  flow across the bilayer. The study considers a single vesicle in
  background flows such as shear flow (panel (d)) and hydrodynamic
  vesicle-vesicle interactions (panel (e)).}
\label{fig:figure0}
\end{figure}

Recently,~\cite{Fu20} illustrated that the hydrophobic interactions
between Janus particles (JP) in a viscous solvent can be used as a
coarse-grained model to capture the mechanics of an elastic bilayer
membrane of amphiphilic macromolecules such as lipids. Depending on
their total number and geometry, JP suspensions can aggregate to form a
micelle, a patch of bilayer membrane with open ends, and a self-enclosed
bilayer membrane, referred to as a JP vesicle. Using a hybrid continuum
model for the interactions between amphiphilic particles in a viscous
solvent and with a boundary integral formulation,~\cite{Fu20} showed
that the granularity of membrane remodeling, as occurs during fusion and
fission of bilayers, can be accurately captured by the coarse-grained
model.  

In the present work, we extend the hybrid continuum model in~\cite{Fu20}
for the Janus suspension to incorporate the collective hydrodynamics of
a Janus suspension under various flows. The JP vesicles in our
simulations replicate  well-know vesicle hydrodynamics such as
tank-treading and inter-leaflet slippage in a shear flow and migration
in Poiseuille flow from continuum models. Furthermore we use the hybrid
continuum model to investigate permeability and rupture of a bilayer
membrane due to an imposed flow.  Finally, we compare a pair of
interacting JP vesicles with a continuum model of a pair of vesicles.

\cite{Brandner2019} used the coarse-grained force field with a lattice
Boltzmann molecular dynamics to simulate the hydrodynamics of a
nano-sized vesicle under a shear flow.  In MD simulations, the
hydrodynamic interactions for the solvent phase are often approximated
by an implicit solvent coarse-grained model.  In the present work, the
hydrodynamic interactions come from the mobility problem for the Stokes
equations for the incompressible, viscous solvent.  

We require a numerical method to avoid unphysical contact between rigid
Janus particles. Optimization-based contact methods introduce
constraints, such as enforcing a non-positive space-time interference
volume (\cite{lu-rah-zor2017, Lukas19, yan-cor-mal-vee-she2020}). These
methods do not introduce stiffness, but do require solving potentially
expensive nonlinear complementarity problems at each time step.
Repulsion-based contact methods, which we employ, introduce an
artificial repulsion force that increases in strength as two particles
approach one another (\cite{glo-pan-hes-jos-per2001, fen-mic2004,
kab-qua-bir2018}). Strong repulsion forces can introduce numerical
stiffness, but with the presence of the hydrophobic forces, we maintain
contact-free suspensions with a relatively weak non-stiff repulsive
force.

The paper is organized as follows. In \S~\ref{sec:governing_eqs} we
present the formulation for the Janus particles in a viscous fluid in
the zero-Reynolds number regime (Figure \ref{fig:figure0}a). In
\S~\ref{sec:IEM} we extend the hybrid continuum model for a Janus
suspension to include the effects of a far-field flow via the mobility
problem formulation. In \S~\ref{results} we validate our model and
present simulation results for a single JP vesicle (Figure
\ref{fig:figure0}d) and for a pair of JP vesicles (Figure
\ref{fig:figure0}e) under various flowing conditions. Finally we provide
discussion and outlook for future directions in \S~\ref{conclusion}.


\section{Governing Equations\label{sec:governing_eqs}}
\subsection{\label{mobility}Mobility Problem}
The objective of this work is to study the hydrodynamics of JP vesicles
in background flows. We consider an $N_b$-many body collection of JP
suspended in a two-dimensional unbounded domain $\Omega$. The boundary
of each particle is denoted by $\Gamma_i$ so that $\bd \Omega = \Gamma_1
\cup \Gamma_2 \cup \cdots \cup \Gamma_{N_b}$ (Figure
\ref{fig:figure0}a). Assuming the inertial terms are negligible, the
governing equations are
\begin{alignat}{3}
  -\mu \Delta \uu + \nabla p &= \mathbf{0}, 
    && \xx \in \Omega, \\
  \nabla\cdot \uu &= 0, \qquad && \xx \in \Omega, \\
  \uu - \uu_\infty &\to \mathbf{0}, && |\xx| \to \infty,
\end{alignat}
where $\uu$ is the velocity, $p$ is the
pressure, $\uu_\infty$ is the background flow, and $\mu$ is the constant viscosity. 
Since each particle $\Gamma_i$ with centre $\aa_i$ is a rigid body, its velocity satisfies 
\begin{align}
  \vv(\xx) = \vv_i + \omega_i (\xx - \aa_i)^\perp, \quad 
    \xx \in \Gamma_i,
\end{align}
where $\vv_i$ is its translational velocity and $\omega_i$ is its
angular velocity. Here, $\langle x, y \rangle^{\perp} = \langle -y, x
\rangle$. Therefore, the no-slip boundary condition on each particle is
\begin{align}
\label{eq:rigid_bc}
  \uu(\xx) = \vv_i + \omega_i (\xx - \aa_i)^\perp, \quad
    \xx \in \Gamma_i.
\end{align}
To determine the translational and angular velocities of each particle,
we define imposed forces $\FF_i$ and torques $T_i$ acting on each
particle. Since the small particles are inertialess, force and torque
balance gives 
\begin{alignat}{2}
  \label{eq:force}
  \FF_i &- \int_{\Gamma_i} \ssigma \cdot \nnu \, \dif s = \mathbf{0},
  && i=1,\ldots,N_b,\\
  \label{eq:torque}
  T_i &- \int_{\Gamma_i} (\xx - \aa_i)^\perp \cdot 
    (\ssigma \cdot \nnu) \, \dif s = 0, \qquad && i=1,\ldots,N_b,
\end{alignat}
where $\ssigma = -p \mathbf{I} + \mu \left(\nabla \uu + \nabla \uu^T
\right)$ is the hydrodynamic stress tensor (pressure tensor) and
$\nnu_i$ is the particle outward normal. The process of finding the
translational and angular velocities given the forces and torques is
referred to as the mobility problem.

\subsection{Imposed Forces}
The imposed forces and torques contain two parts: hydrophobic attraction
and repulsion. The hydrophobic attraction potential was introduced
by~\cite{Fu20} and is responsible for forming particle aggregates that
sequester their hydrophobic surface regions (Figure \ref{fig:figure0}b).
We model hydrophobic attraction by solving the screened Laplace equation
boundary value problem
\begin{alignat}{2}
  \label{eq:SL}
-\rho^2 \Delta u + u &=0,            && \xx \in \Omega,\\
\label{eq:SLbc}
u(\xx) &= f_i(\xx),\qquad  && \xx \in \Gamma_i,\; i=1,\ldots,N_b, \\
\label{eq:SLff}
u &\to 0,                          &&|\xx| \to \infty,
\end{alignat}
where $0 \leq f_i \leq 1$ is a material label with $f_i = 0$,
respectively $1$, representing hydrophilic, respectively hydrophobic,
portions of the surface. We assume that both $f_i$ and $\Gamma_i$ are
smooth. The parameter $\rho > 0$ is the decay length of attraction. The
forces and torques of attraction are 
\begin{align}
  \label{eq:hydrophobicAttraction}
  \FF_i^{\text{hydro}} = \int_{\Gamma_i} {\bf T}\cdot \nnu \, \dif s, 
    \quad 
  T_i^{\text{hydro}} = \int_{\Gamma_i} (\xx - \aa_i)^{\perp} \cdot ({\bf T} \cdot \nnu) \dif s,
\end{align}
where
\begin{align}
  \label{eq:stress}
\mathbf{T}
= \gamma\rho^{-1}u^2 \mathbf{I} + 2\rho\gamma \left(\tfrac{1}{2}|\nabla
  u|^2 \mathbf{I} - \nabla u  \nabla u^T\right)
\end{align}
is the hydrophobic stress tensor and $\gamma > 0$ is
the interfacial
tension. 

The second part of the imposed forces and torques comes from repulsion
between proximal particles. Given a pair particles indexed with $i$ and
$j$, we find the two points $\tilde{\xx}_i \in \Gamma_i$ and
$\tilde{\xx}_j \in \Gamma_j$ that are closest to one another (Figure \ref{fig:figure0}a). We then
define the repulsion force and torque
\begin{align}
  \label{eq:REPULforce}
  \FF_i^{\text{repul}} &= \sum_{j \neq i} 
    \frac{\tilde{\xx}_i - \tilde{\xx}_j}
    {|\tilde{\xx}_i - \tilde{\xx}_j|} 
    P'(|\tilde{\xx}_i - \tilde{\xx}_j|), \\
  \label{eq:REPULtorque}
  T_i^{\text{repul}} &= \sum_{j \neq i} 
    (\tilde{\xx}_i - \aa_i)^{\perp} \cdot 
    \frac{\tilde{\xx}_i - \tilde{\xx}_j}
    {|\tilde{\xx}_i - \tilde{\xx}_j|} 
    P'(|\tilde{\xx}_i - \tilde{\xx}_j|).
\end{align}
The repulsion profile $P(r)$ is set to zero for distances $r$ larger than a
repulsion length scale $\rho_0$. As such,~\eqref{eq:REPULforce}
and~\eqref{eq:REPULtorque} ignore particles outside a
$\rho_0$-tubular neighborhood of $\Gamma_i$. For $0 \leq r < \rho_0,$ we use
$P(r) = M(1 - \sin(r/\rho_0))$ where $M$ is sufficiently large to
prevent particle collisions. Then, the total imposed force and torque
are
\begin{align}
  \FF_i = \FF_i^{\text{hydro}} + \FF_i^{\text{repul}},\quad
  T_i = T_i^{\text{hydro}} + T_i^{\text{repul}}, \qquad
  i=1,\ldots,N_b.
\end{align}

\subsection{Time Marching}
By solving the mobility problem, we obtain translational and angular
velocities of the $N_b$-body system. 
A second-order Adams-Bashforth
scheme updates the particle positions and orientations. By including the
repulsion force~\eqref{eq:REPULforce}, particle collisions are avoided even
when using a relatively large time step.


\section{Integral Equation Method}
\label{sec:IEM}
Computing the hydrophobic attraction potential and the particle forces
and torques requires the solution of elliptic partial differential
equations (PDEs) in an unbounded complex domain. We recast both these
PDEs as boundary integral equations (BIEs). We discretize each BIE at
$N$ points on each of the $N_b$ particles with a collocation method.
Integrals that are smooth are computed with the spectrally-accurate
trapezoid rule, and nearly-singular integrals, caused by close contact
between two particles, are computed with a high-order
interpolation-based quadrature rule (\cite{qua-bir2014}). After
discretizing and applying quadrature, the resulting linear system is
solved with matrix-free GMRES, and we guarantee that the number of GMRES
iterations is mesh-independent by using second-kind BIEs.

\subsection{HAP Boundary Integral Equation}
Similar to our previous work (\cite{Fu20}), we represent the HAP as a
double-layer potential
\begin{align}
  \label{eq:HAP_DLP}
  u(\xx) = \frac{1}{2\pi} \int_{\bd\Omega} \pderiv{}{\nnu_\yy}
    K_0 \left(\frac{|\xx - \yy|}{\rho}\right) \sigma(\yy) \, \dif s_\yy,
    \quad \xx \in \Omega,
\end{align}
where $K_0$ is the zeroth-order modified Bessel function of the first
kind and the integral is taken in the sense of principle value whenever $\xx \in \partial \Omega$. 
By requiring that the density function $\sigma$ satisfies the
second-kind integral equation
\begin{equation}
\label{eq:screenedSKIE}
  f(\xx) = \frac{1}{2}\sigma(\xx) + 
    \frac{1}{2\pi}\int_{\bd\Omega} \pderiv{}{\nnu_\yy}
    K_0 \left(\frac{|\xx - \yy|}{\rho}\right) \sigma(\yy) \, \dif s_\yy,
    \quad \xx \in \bd\Omega,
\end{equation}
the HAP double-layer potential~\eqref{eq:HAP_DLP} satisfies the screened
Laplace equation~\eqref{eq:SL}--\eqref{eq:SLff}. After
discretizing~\eqref{eq:screenedSKIE}, the result is an $NN_b \times
NN_b$ linear system that is solved with block-diagonal preconditioned
GMRES.

To calculate the hydrophobic force and torque, the gradient of the
double-layer potential~\eqref{eq:HAP_DLP} must be computed on the
boundary of each particle. The resulting integrands are singular, and
specialized quadrature would be necessary to approximate such integrals.
Alternatively, in Section~\ref{subsec:calculating_force}, we show how
the force and torque calculations can be expressed in terms of
non-singular integrals.

\subsection{Mobility Problem Boundary Integral Equation}
Following previous work of~\cite{Lukas19}, we use the velocity
representation of~\cite{pow-mir1987}. In particular, we write the
velocity as the sum of a double-layer potential and $N_b$-many Stokeslets and
rotlets
\begin{align}
  \label{eq:velocity}
  \uu(\xx) = \uu_\infty(\xx) + \DD[\eeta](\xx) + 
    \sum_{i=1}^{N_b} \left(\SS(\xx,\aa_i) \FF_i + 
    \RR(\xx,\aa_i) T_i\right), \quad \xx \in \Omega.
\end{align}
The double-layer potential is
\begin{align}
  \DD[\eeta](\xx) = \frac{1}{\pi} \int_{\bd\Omega} 
    \frac{\rr \cdot \nnu}{|\rr|^2} \frac{\rr \otimes \rr}{|\rr|^2}
    \eeta(\yy) \, \dif s_\yy,
\end{align}
where $\rr = \xx - \yy$ and $\rho = |\rr|$. The Stokeslet and rotlets
centred at $\aa_i$ are
\begin{align}
  \SS(\xx,\aa_i)\FF_i &= \frac{1}{4\pi} \left(-\log |\rr| + 
    \frac{\rr \otimes \rr}{|\rr|^2}\right) \FF_i, \\
  \RR(\xx,\aa_i)T_i &= \frac{1}{4\pi} \frac{\rr^\perp}{|\rr|^2} T_i,
\end{align}
respectively, where $\rr = \xx - \aa_i$. The
Stokeslet is torque-free and has force $\FF_i$ while the rotlet is
force-free and has torque $T_i$. Therefore, the
velocity~\eqref{eq:velocity} satisfies the total force~\eqref{eq:force}
and torque~\eqref{eq:torque} conditions if the double-layer potential
$\DD[\eeta]$ is force- and torque-free.  Matching the limit
of~\eqref{eq:velocity} with the rigid body motion, and imposing that
$\DD[\eeta]$ is force- and torque-free, the density function $\eeta$,
translational velocity $\vv_i$, and angular velocity $\omega_i$ satisfy
\begin{alignat}{3}
  \nonumber
  \vv_i + \omega_i (\xx - \aa_i)^\perp &= \uu_\infty(\xx)
    -\frac{1}{2} \eeta(\xx) + \DD[\eeta](\xx) \\
  \label{eq:SKIE}
    + \sum_{j=1}^{N_b} &
    \left(\SS(\xx,\aa_j) \FF_j + \RR(\xx,\aa_j) T_j\right),
    \quad &&\xx \in \Gamma_i,\: i=1,\ldots,N_b, \\
  \label{eq:mobility1}
  \int_{\bd\Gamma_i} \eeta \cdot \nnu_i \, \dif s &= {\bf 0}, 
  &&i = 1,\ldots,N_b, \\
  \label{eq:mobility2}
  \int_{\bd\Gamma_i} \eeta\times(\xx-\aa_i)^\perp \cdot \nnu_i \, \dif s &= 0,
  &&i = 1,\ldots,N_b.
\end{alignat}

After discretizing and applying appropriate quadrature rules, the result
is a $(2NN_b + 3N_b) \times (2NN_b + 3N_b)$ linear system that we solve
with block-diagonal preconditioned GMRES. Other BIE formulations of the
mobility problem use single-layer potentials (\cite{cor-gre-rac-vee2017,
rac-gre2016}) or a combination of single- and double-layer potentials
(\cite{cor-vee2018}).

We have validated our solver for~\eqref{eq:SKIE}--\eqref{eq:mobility2}
using a single elliptical particle suspended in a background shear flow
(see \S~\ref{sec:ves_in_shear}). Hydrophobic attraction and repulsion
are zero for a single particle (see \cite{Fu20} equation 2.13). The
angle of the ellipse's major axis coming from the integral equation
method agrees with the theoretical, Jeffery orbit time-course
(\cite{jef1922}). 
  
%


\subsection{Main Theoretical Result: Calculating the Hydrophobic Force}
\label{subsec:calculating_force}
Once \eqref{eq:screenedSKIE} has been solved for $\sigma$, we need to
evaluate the integrals~\eqref{eq:hydrophobicAttraction} which are the
HAP forces and torques. These integrals involve the
stress~\eqref{eq:stress} which contains a singular integral for the
gradient of the double-layer potential. To avoid singular integrals, we
first define
\begin{align}
  \label{eq:uidef}
  v_i(\xx) = u(\xx) - u_i(\xx),
\end{align}
where
\begin{align}
  u_i(\xx) &= \frac{1}{2\pi} \int_{\Gamma_i} \pderiv{}{\nnu_\yy}
    K_0 \left(\frac{|\xx - \yy|}{\rho}\right) \sigma(\yy) \, \dif s_\yy,
    \quad \xx \in \mathbb{R}^2.
\end{align}
That is, $v_i(\xx)$ is the double-layer potential~\eqref{eq:HAP_DLP}
with $\Gamma_i$ excluded from $\bd\Omega$. Having defined $v_i$, we prove
\begin{theorem}
\begin{align}
  \label{eq:recipforcetorque}
  \FF_i^{\mathrm{hydro}} = \int_{\Gamma_i} \JJ_i \,\dif s,\quad
  T_i^{\mathrm{hydro}}    = \int_{\Gamma_i} 
    (\xx - \aa_i)^{\perp} \cdot \JJ_i  \,\dif s,
\end{align}
where
\begin{align}
  \label{eq:jumpstress1}
  \JJ_{i} = 2\gamma\rho^{-1} \sigma v_i \nnu + 
    2\gamma\rho \frac{d\sigma}{ds} \frac{dv_i}{ds} \nnu -
    2\gamma\rho \frac{d\sigma}{ds} \frac{dv_i}{d\nnu} \ttau.
\end{align}
\end{theorem}
The symbols $\ttau$ and $\frac{d}{ds}$ are the unit tangent and arc
length derivative for $\Gamma_i$, respectively (Figure
\ref{fig:figure0}a.) The result is valid for any smooth particle shape
or boundary condition. The advantage of
using~\eqref{eq:recipforcetorque} over
using~\eqref{eq:hydrophobicAttraction} is that the components $\sigma$,
$v_i$, $\frac{d\sigma}{ds}$, and $\frac{dv_i}{ds}$ of $\JJ_i$ are smooth
functions, whereas the components of \eqref{eq:stress} are singular
integrals. 


To prove \eqref{eq:recipforcetorque}, let 
\begin{align*}
  {\bf T} &=
    \mathbf{S}(u_i,u_i)
  +(\mathbf{S}(u_i,v_i)
  +\mathbf{S}(v_i,u_i))
  +\mathbf{S}(v_i,v_i) \\
  &= {\bf T}_1 + {\bf T}_2 + {\bf T}_3
\end{align*}
where we introduce the bilinear form
\begin{equation}
\label{eq:Tsplit}
\mathbf{S}(u,v)
=  \gamma\rho^{-1} uv {\bf I} + \gamma\rho \nabla u \cdot \nabla v {\bf I} - 2 \gamma \rho \nabla u \nabla v^T .
\end{equation}
Using the fact that $u$,  $u_i$, and $v_i$ solve the screened Laplace
equation~\eqref{eq:SL},  and that $\mathbf{T}$, $\mathbf{T}_j$, $j = 1,
2, 3$ are symmetric, it is straightforward to verify that 
  \begin{equation}
    \label{eq:decompdivfree}
    \nabla \cdot {\bf T}_j = 0, \quad
    \nabla \cdot ((\xx-\aa_i)^{\perp} \cdot {\bf T}_j) = 0, \quad j = 1, 2, 3.
  \end{equation}

Let $U_i$ be the interior of the particle indexed by $i$. For $\xx_0
\in \Gamma_i$ and an arbitrary function $g(\xx)$, the notation
\begin{align}
  \jump{g}(\xx_0) = \lim_{\substack{\xx \to \xx_0 \\ \xx \in U_i^c}}g(\xx)  - 
                    \lim_{\substack{\xx \to \xx_0 \\ \xx \in U_i}}g(\xx),
\end{align}
denotes the jump of the limits of $g(\xx)$ taken from the outside to
the inside of $\Gamma_i$.
\begin{lemma}
\begin{align}
  \label{eq:prejump}
  \FF_i^{\mathrm{hydro}} = \int_{\Gamma_i} \jump{{\bf T}_2  \nnu}  \, \dif s,\quad
  T_i^{\mathrm{hydro}} = \int_{\Gamma_i} (\xx - \aa_i)^{\perp} \cdot \jump{{\bf T}_2 \nnu} \, \dif s.
\end{align}
\end{lemma}
\begin{proof}
To show~\eqref{eq:prejump}, we expand \eqref{eq:hydrophobicAttraction} as
\begin{align}
  \FF_i = \int_{\Gamma_i} {\bf T}_1\nnu + 
    {\bf T}_2\nnu + {\bf T}_3\nnu\,\dif s.
\end{align}
By~\eqref{eq:uidef} and~\eqref{eq:decompdivfree}, we have that $u_i$ is
smooth and $\nabla \cdot {\bf T}_1 = \mathbf{0}$ in $\mathbb{R}^n
\setminus U_i$. Similarly, $v_i = u - u_i$ is smooth and $\nabla \cdot
{\bf T}_3 = \mathbf{0}$ in $U_i$. By the divergence theorem,  
\begin{align}
  \int_{\Gamma_i}  {\bf T}_1\nnu \,\dif s
  = -\int_{\mathbb{R}^n \setminus U_i} \nabla \cdot {\bf T}_1 \,\dif \xx = \mathbf{0},\quad
    \int_{\Gamma_i}  {\bf T}_3 \nnu\,\dif s
  = \int_{U_i} \nabla \cdot {\bf T}_3 \,\dif \xx = \mathbf{0}.
\end{align}
Finally, $u_i$ and $v_i$ are smooth and $\nabla \cdot {\bf T}_2= 0$ in $U_i$. This gives
\begin{align}
  \mathbf{0} = \int_{U_i} \nabla \cdot {\bf T}_2 \, \dif \xx = -\int_{\Gamma_i}  
    ({\bf T}_2 \nnu)^-\,\dif s,
\end{align}
where the superscript denotes the limit taken from in $U_i$.  
Combining the above gives the first equation in \eqref{eq:prejump}.
The argument for the second equation in \eqref{eq:prejump} is identical. 
\end{proof}

We have the following jump relations for \eqref{eq:uidef}: 
\begin{equation}
\label{eq:DLjump}
\jump{u_i} = \sigma, \quad
\jump{v_i} = 0, \quad
\jump{\nabla u_i} = \frac{\dif \sigma}{\dif s}\ttau,\quad
\jump{\nabla v_i} = \mathbf{0},
\end{equation}
on $\Gamma_i$ (see, e.g. \cite{KlBaGrON13}). Therefore,
\begin{align*}
  \jump{{\bf T}_2\nnu}   &= \jump{\mathbf{S}(u_i,v_i)\nnu  +\mathbf{S}(v_i,u_i)\nnu} \\
  &= \jump{( 2\gamma\rho^{-1} u_i v_i \mathbf{I} + 2\gamma\rho \nabla u_i \cdot \nabla v_i \mathbf{I} 
- 2\gamma\rho \nabla u_i  \nabla v_i^T - 2\gamma\rho \nabla v_i \nabla u_i^T)  \nnu}\\
&= 2\gamma\rho^{-1} \sigma v_i \nnu + 2\gamma\rho \frac{d\sigma }{ds}\frac{dv_i }{ds} \nnu
- 2\gamma\rho \ \frac{d\sigma }{ds} \frac{dv_i }{d\nu} \ttau = \JJ_i.
\end{align*}
Combining this with~\eqref{eq:prejump} gives~\eqref{eq:recipforcetorque}
as required.

\section{\label{results}Numerical Results}

\begin{figure}
\centering
\includegraphics[width=11.5cm]{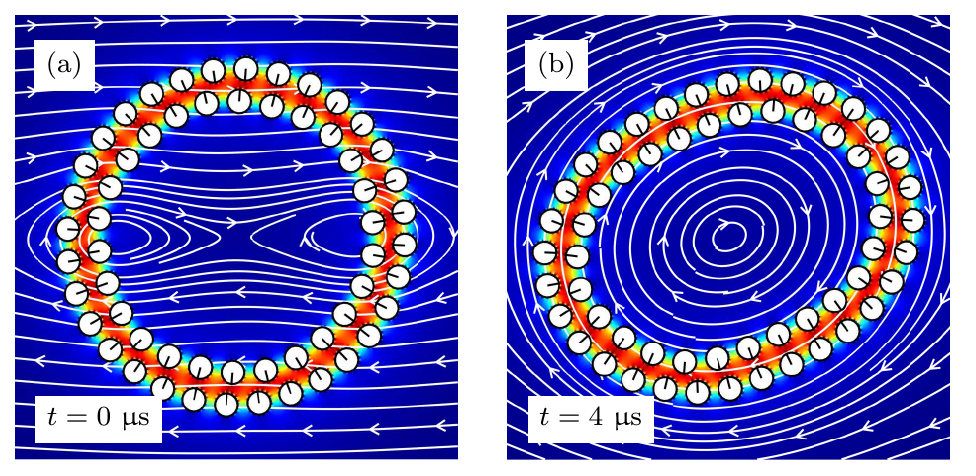}
  \caption{\label{figure3} 
  JP vesicles undergo tank-treading in shear flow. 
  In panel (a),  the initial, 58-body vesicle is circular and the black arrows point in the direction of the
  hydrophobic side of the JP.
  In panel (b), the JP suspension rotates and deforms with the shear flow.
  The color map is for the solution $u$ of \eqref{eq:SL}. 
  The white curves are the streamlines of $\uu.$ 
   The shear rate is $\chi=0.0025$.}
\end{figure}

\subsection{Model Parameters}
\cite{Fu20} studied physical quantities for static JP vesicles for
various particle shapes. In the present study, we fix the particle shape
and vary the background flows. Specifically, the particles are circular
disks with diameter $l_0 = 2.5$~nm. The boundary
conditions~\eqref{eq:SLbc} are $f_i(\xx) = \tfrac{1}{2}(1 + \cos
\theta_i)$ where $\theta_i$ is the angle between $\xx - \aa_i$ and
$\dd_i$, where the vectors $\aa_i$ and $\dd_i$ are the particle centre
and director respectively. The particle diameter is the thickness of
monolayers and the director points in the direction of the hydrophobic
side of the JP. 

We use $\rho = 2 l_0$ for decay length, $\rho_0 = 0.2l_0$ for repulsion
length, $M=4.0$~$k_BT$ for repulsion strength, $\gamma = \text{4.1 pN
nm}^{-1}$ for interfacial tension, and $\mu = \text{1 cP} = \text{1 pN
ns nm}^{-2}$ for viscosity of room-temperature water. We
nondimensionalize the problem through the change of variables 
$t \mapsto t \text{ ns}$,
$\xx \mapsto \xx \text{ nm}$, 
$\uu \mapsto \uu \text{ nm ns}^{-1}$, and 
$p \mapsto p \text{ pN nm}^{-2}$.
The numerical time step size is $\Delta t=0.2$.

To reach consistent simulation outcomes, we first solve for a baseline
JP vesicles that is suitably close to equilibrium. We start with an
assumed configuration of $N=58$ JP in the form of two, circular,
apposing monolayers of about $8$ nm in radius. The norms of the translational and
rotational velocities vanish exponentially with an approximate decay
rate $\sim 4.6$ $\upmu$s$^{-1}$. An equilibrium configuration is
therefore rapidly attained. This equilibrium configuration serves as the
initial data in the subsequent background flow simulations.

\subsection{Tank-Treading Vesicles}
\begin{figure}
\begin{center}
\includegraphics[width=13.5cm]{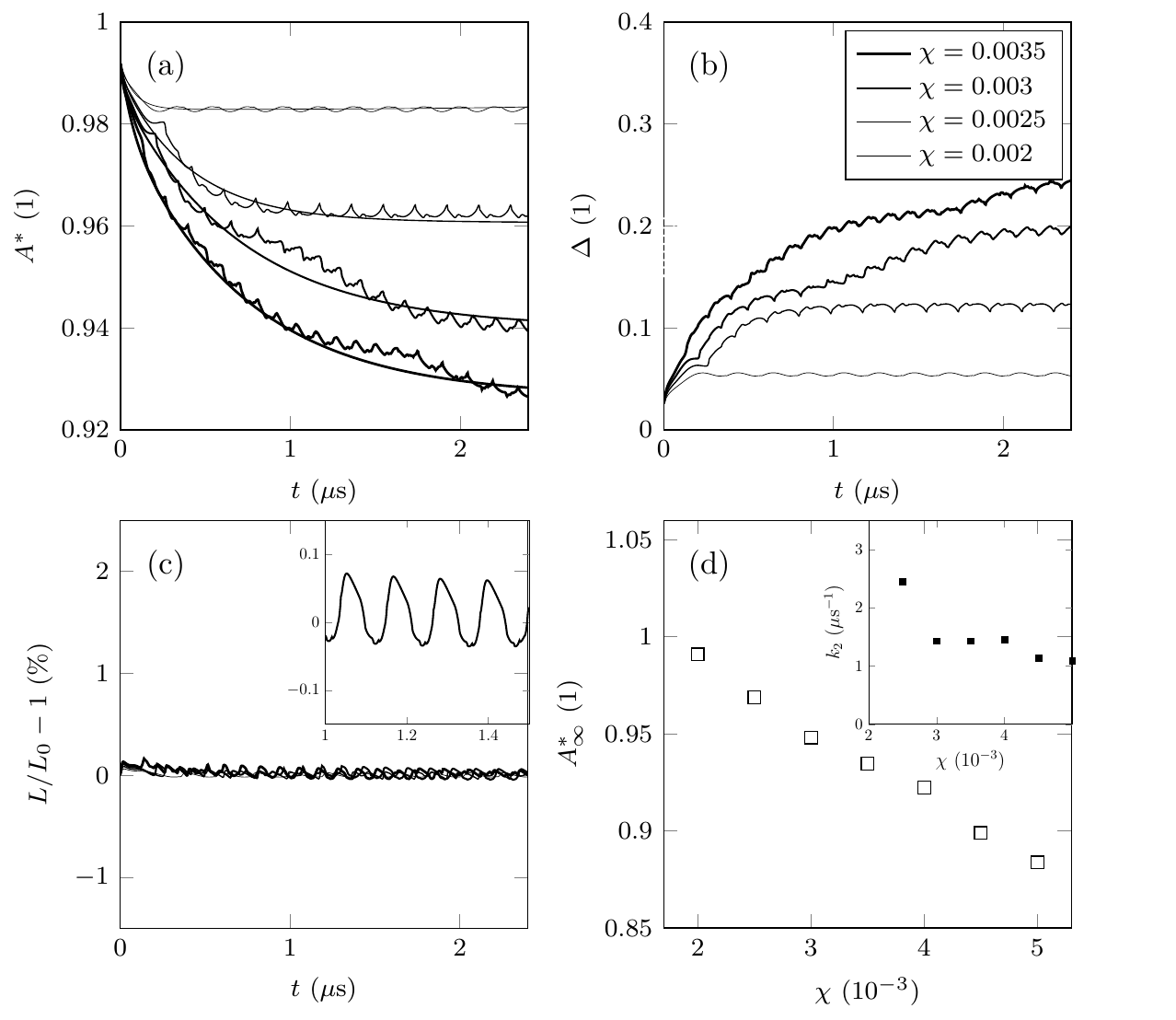}
\end{center} 
  \caption{\label{figure4} JP vesicles lose enclosed area but conserve
  length in a shear flow. Panel (a) gives the reduced area over time for
  four different shear rates. The monotonic curves are a two-exponential
  best fit. There is a commensurate increase in excess length (panel
  (b)), but the arc length of the mid plane curve is more or less
  constant for all shear rates (panel (c)). The inset in panel (c) is
  for shear rate 0.003. The legend in panel (b) applies to panels
  (a)--(c). Panel (d) plots the steady-state reduced area and the decay
  rate $k_2$ (inset) coming from the fitting data in panel (a).}
\end{figure}

\subsubsection{Vesicle in a Shear Flow}
\label{sec:ves_in_shear}
Our simulation studies begin by showing, for the first time, that a JP
suspension with hydrophobic attraction behaves as a tank-treading
vesicle (\cite{Finken08, Shaqfeh11}). The centroid of the baseline,
$58$-body JP vesicle lies at the origin and the background shear flow 
\begin{align}
  \uu_{\infty}(\xx) = \dot\gamma (\xx \cdot \mathbf{e}_y) \mathbf{e}_x,
\end{align} 
is applied for shear rate $\dot\gamma$, and orthogonal unit vectors
$\mathbf{e}_x$ and $\mathbf{e}_y$ for the horizontal and vertical
directions, respectively. We use the dimensionless shear rate $\chi =
\dot \gamma$ (s$^{-1}$) ns.

\begin{figure}
\begin{center}
  \includegraphics[width=11cm]{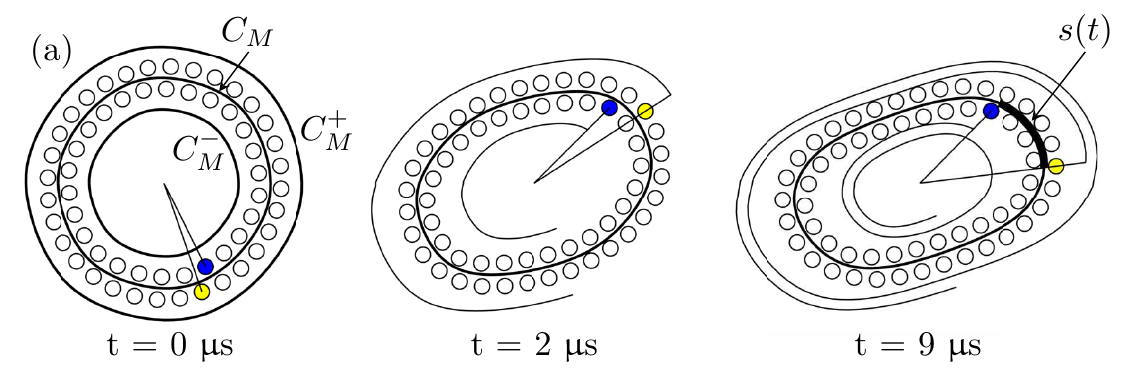}\\
  \includegraphics[width=11cm]{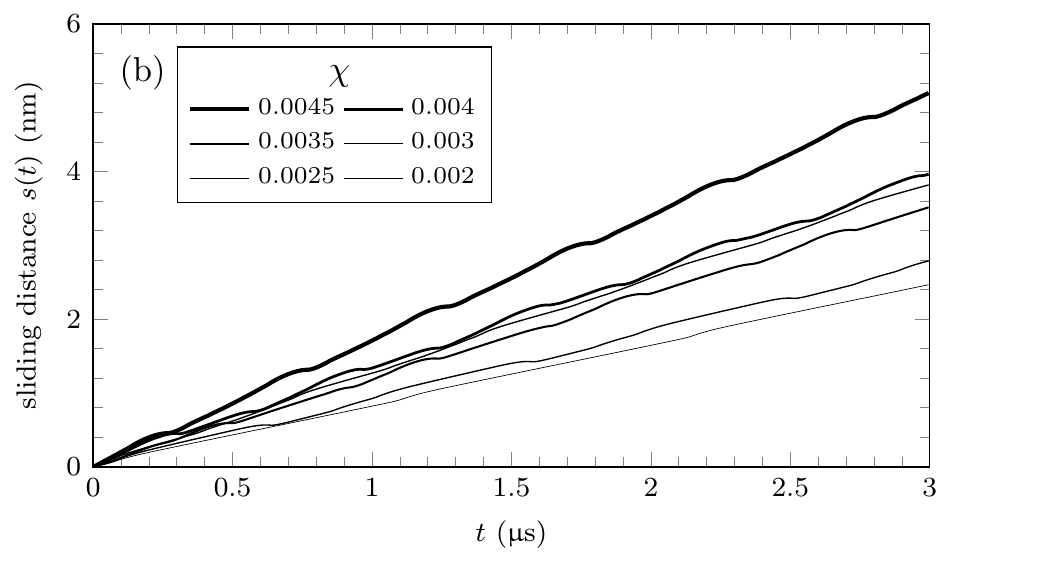}
\end{center} 
\caption{\label{figure5} Inter-monolayer slip is present in the JP
  vesicle tank-treading motion at all shear rates. Panel (a) tracks a
  pair of particles (blue and yellow). The shear rate is $\chi=0.005$
  and in the right figure, the yellow particle will complete two and a
  quarter revolutions in the same time the blue particle completes two
  revolutions. The curves in panel (b) are the distance the outer
  leaflet has slid past the inner leaflet. The slopes of the curves
  give the slip velocity. With the exception of $\chi = 0.0035$, the
  slip velocity is generally monotonic in the shear rate.}
\end{figure}

Figure~\ref{figure3} shows snapshots of the JP vesicle in the shear flow
with $\chi=0.0025$. Under the background flow, the rigid-body boundary
condition~\eqref{eq:rigid_bc} causes the particle suspension to
elongate. This perturbation disrupts the preferred particle orientations
and exposes the hydrophobic core to bulk water (Figure~\ref{figure3},
red region). In response, hydrophobic
attraction~\eqref{eq:hydrophobicAttraction} causes the particles to
reorient and form a somewhat elliptically-shaped suspension. Panel (a)
shows the initial configuration and panel (b) shows the later,
fully-formed, clockwise tank-treading motion. The JP suspension
maintains its bilayer structure throughout the simulation. 

To extract physical quantities, let $A^* = 4\pi A/L^2$ be the reduced
area and $\Delta=L/\sqrt{A/\pi} - 2\pi$ the excess length of the bilayer
structure (\cite{Finken08}). Here, $A$ is the enclosed area and $L$ is
the total length of the JP vesicle (Figure \ref{fig:figure0}a).
Figure~\ref{figure4} shows the evolution of the area and length for
various shear rates. Panels (a) and (b) show that $A^*$ decreases and
$\Delta$ increases with time, respectively, and that the rate of
decrease/increase grows with shear rate. The total arc length, however,
remains constant for all time for all four shear rates (panel (c)). We
conclude that the JP vesicle loses area and that the bilayer behaves as
a permeable membrane. 

Since the starting configuration is nearly circular, the reduced area
decreases from an initial value close to $1$ and tends to a steady-state
value $A^*_{\infty}$. The data in panel (a) are fit to the model $A^* =
(A^*_{\infty}-a_1-a_2) + a_1 \exp(-k_1t) + a_2 \exp(-k_2t),$ $0 < k_2 <
k_1$. Panel (d) shows that $A^*_{\infty}$ decreases with the shear rate.
The JP vesicle achieves a steady-state reduced area earlier when the
shear rate is low (inset), but the decay rate $k_2$ is more or less
constant for higher shearer rates. 

The oscillations in the data of Figure~\ref{figure4} are due to the
granularity of the JP vesicle. The inset of Figure~\ref{figure4}c zooms
in on the arc length data for the shear rate $\chi = 0.003$. It shows
that the oscillations are smooth and well-resolved by our second-order
Adams-Bashforth scheme.

We point out that the range of values for $\chi$ where we measured for
tank-treading correspond to shear rates $\dot \gamma = \mathcal{O}(10^6$
s$^{-1}$) which gives fluid velocities $\mathcal{O}($m s$^{-1}$$)$ in
the vicinity of the vesicle. While large, these orders of magnitude are identical to
ones used in prior MD studies (\cite{Brandner2019}) and are a
consequence of the fact that larger shear rates are required to produce
the viscous stresses needed to appreciably deform smaller vesicles. 

\subsubsection{Inter-Monolayer Friction}
We observe inter-monolayer slip in the tank-treading, JP vesicle at all
shear rates. Since they are not bound, the two leaflets of the vesicle
are able to slide past one another. 
Monolayer slip effects have been incorporated in continuum models (\cite{sch-vla-mik2010}).
In the present setting, slip is limited by viscous
friction of the aqueous gaps between particles and by the constant
unbinding and binding of particles pairs in apposing leaflets.

Figure~\ref{figure5}a illustrates inter-monolayer slip by tracking the distances
traveled by a pair of particles along the midplane curve. In the left image, the blue and yellow
particle lie next to each other. In the right-most panel, the yellow
particle has traveled farther than the blue particle. This suggests that
the outer tangential velocity, obtained by projecting the velocity of the
outer leaflet onto the midplane curve, 
is larger than that of the inner leaflet. 

From the data, we obtain an inter-monolayer friction coefficient 
\begin{align}
  b =  \frac{\langle F \rangle}{\langle L   U \rangle} ,
\end{align}
where $F$ is the tangential force jump, $L$ is the length of the
midplane $C_M$ (Figure~\ref{figure5}a), and $U$ is the slip velocity.
The time average $\langle \cdot \rangle$ is necessary to avoid division by zero
whenever slip velocity vanishes.  

\begin{table}
\caption{Friction Coefficients}
\centering
\begin{tabular}{c c c c c c c c }
 $\dot\gamma$  (ns$^{-1}$) & 0.0020   &  0.0025 &  0.0030 &  0.0035 &  0.0040 & 0.0045 & 0.0050  \\
\hline                    
$b$ (pN ns nm$^{-3}$)    & 0.43 & 1.19 &   0.42  & 0.69 &   0.97  &   0.80  &  1.07 \\ 
\hline    
\end{tabular} 
\label{table1}
\end{table}

The tangential force jump $F$ equals the 
tangential shear force on the outer
leaflet minus the tangential shear force on the inner leaflet. 
To calculate  $F$, we first let 
\begin{align}
F_h = \int_{\mathcal{C}_M^+} 
  \ttau \cdot \sigma \cdot \nnu \,\dif s - \int_{\mathcal{C}_M^-} 
 \ttau \cdot \sigma \cdot \nnu \,\dif s 
\end{align}
where $C_M^+$ and $C_M^-$ are target curves obtained by projecting the
midplane curve $C_M$  a distance $h$ outward,
respectively inward, along its unit normal vector field (Figure~\ref{figure5}a). We sample $F_h$ for 
$h = \pm 1.3, \pm 1.5,  \pm 1.7, \pm 1.9$ times the particle radius and define $F$ by
extrapolating to zero distance. This avoids integrating along a curve passing directly through the
particles.   

To calculate $U$, we let $u_+(\mathbf{x})$ and $u_-(\mathbf{x})$ 
be the tangential velocity of the outer, respectively inner, leaflet.  Then 
\begin{equation}
U = \frac{1}{L}\int_{C_M}u_+(\mathbf{x})(1 - \delta \kappa) - v_-(\mathbf{x})(1 +  \delta \kappa) \, ds,
\end{equation}
where $ \delta $ is the distance from the leaflet centres to the midplane and $\kappa$ is the curvature. 
The factors $\pm  \delta \kappa$ are needed to project the leaflet velocities, defined on the particle centres,
onto the midplane curve. 
Finally, the function $s(t) = \int_0^t dU $ records the distance one leaflet has
slid past the other. Figure~\ref{figure5}b plots the sliding distance for various shear
rates. 

Table~\ref{table1} provides $b = 0.79 \pm 0.3$ pN ns nm$^{-3}$  
over a range of shear rates
which is in good quantitative agreement with values previously reported in the literature.
Atomistic studies have also considered inter-monolayer slip in lipid bilayers. 
\cite{WuoEd06} and \cite{denOtter2007} reported $b = 0.7 \times 10^6$ Pa m$^{-1}$ s  
$=0.7$ pN ns nm$^{-3}$  
and $b = 2.4$  pN ns nm$^{-3}$ for DPPC membranes simulated by MD, respectively. 
Using a more recent version of the Martini force field, \cite{Zgorski2019} 
gives $b = 5.5$ pN ns nm$^{-3}$ for shear rates 0.4 ns$^{-1}$ and higher.
It is understandable that there is uncertainty in the friction coefficients of
Table~\ref{table1}.  The scatter in our data, however, is fully consistent with that calculated
from MD simulations, c.f. the transient rise in values of Table~\ref{table1} and in \cite{Zgorski2019}, Figure 10
for low shear rates. 
\subsubsection{Membrane Ruptures}
\begin{figure}
\centering
\includegraphics[width=11.32cm]{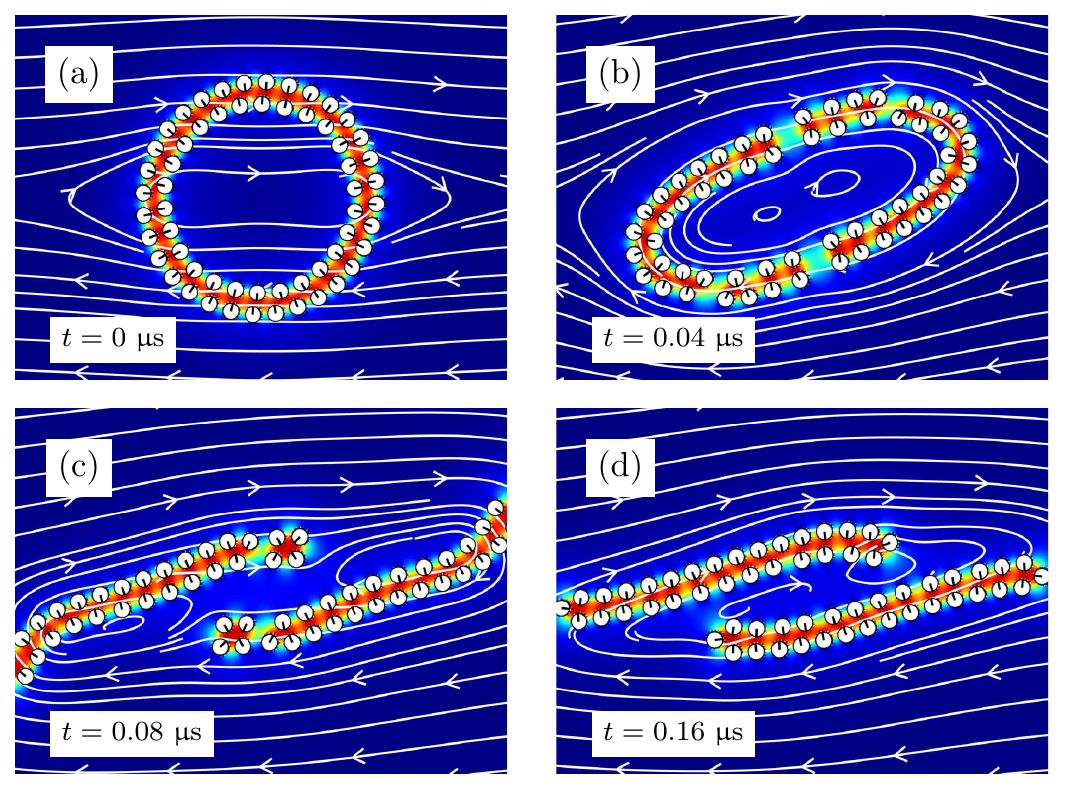}
\\
  \caption{\label{figure8} Tank-treading leads to membrane rupture for
  large shear rates.} 
\end{figure}
A temporary fissure or a complete membrane rupture can occur at large
shear rates. Figure~\ref{figure8} demonstrates how a vesicle can rupture
when suspended in a shear flow. For $\chi = 0.1$, starting with a
circular shape (Figure~\ref{figure8}a), the vesicle is stretched by the
background flow and fissures appear in the bilayer structure in multiple
locations (Figure~\ref{figure8}b). In Figure~\ref{figure8}c, the
ruptured vesicles form two planar micelles which are eventually carried
off by the flow (Figure~\ref{figure8}d).

\subsubsection{Vesicle in a Parabolic Flow}
Finally, we consider the parabolic background flow
\begin{align}
  \uu_\infty = v_{max}\left[ 1 - \left( 
    \frac{\xx \cdot \mathbf{e}_y}{wR_0}\right)^2
    \right]\mathbf{e}_x,
\end{align}
where $v_{max}$ is the flow strength and $w$ determines the shape of the
flow. The parameter $R_0$ is the radius of the JP vesicle at $t=0$ and
$w$ sets the width of the profile. \cite{Kaoui09, cou-kao-pod-mis2008,
dan-vla-mis2009} have shown that the behaviour of a vesicle in this
unbounded flow includes vertical migration, and depending on the flow
rate and reduced area, the steady-state shape can be either a symmetric
parachute or an asymmetric slipper.


Figure~\ref{figure6} shows four configurations for one specific case
where the centroid of the JP vesicle is initially placed slightly above
the $x$-axis. We have marked a pair of particles blue and yellow in the
inner and outer leaflets, respectively, and observe that the deformed JP
vesicle (Figure~\ref{figure6}, $t = 12~\upmu$s) has a counterclockwise
movement and the shape of the vesicle approaches an asymmetric slipper
shape. For this test, the reduced area in the final configuration is
approximately $0.9$ which matches the previous numerical tests
in~\cite{Kaoui09} where a slipper-like shape occurs when the flow
velocity is weak and the reduced area is large. The flow causes the
vesicle, which is initially placed above the axis, to drift downward
where it reaches a steady height of about $1$ nm. 

\begin{figure}
\centering
\includegraphics[width=\textwidth]{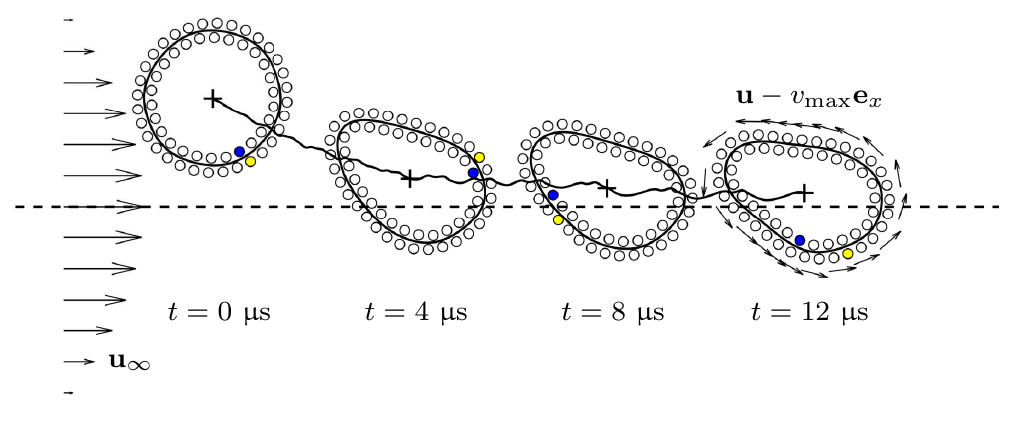}
  \caption{\label{figure6} A JP vesicle placed in a parabolic flow
  drifts toward the central axis. It assumes a slipper shape and
  undergoes tank-treading motion like in the shear flow case. The arrows
  on the left illustrate the background flow and the arrows on the right
  show the velocity relative to the vesicle's moving frame. The thick
  black curve plots the distance from the axis as a function of time
  from left to right. The parameters are $v_{max} = 8$~nm ns$^{-1}$,
  $w=10$, and $R_0=20$ nm.}
\end{figure}


\subsection{Stretching and Permeability}
\begin{figure}
\centering
\includegraphics[width=\textwidth]{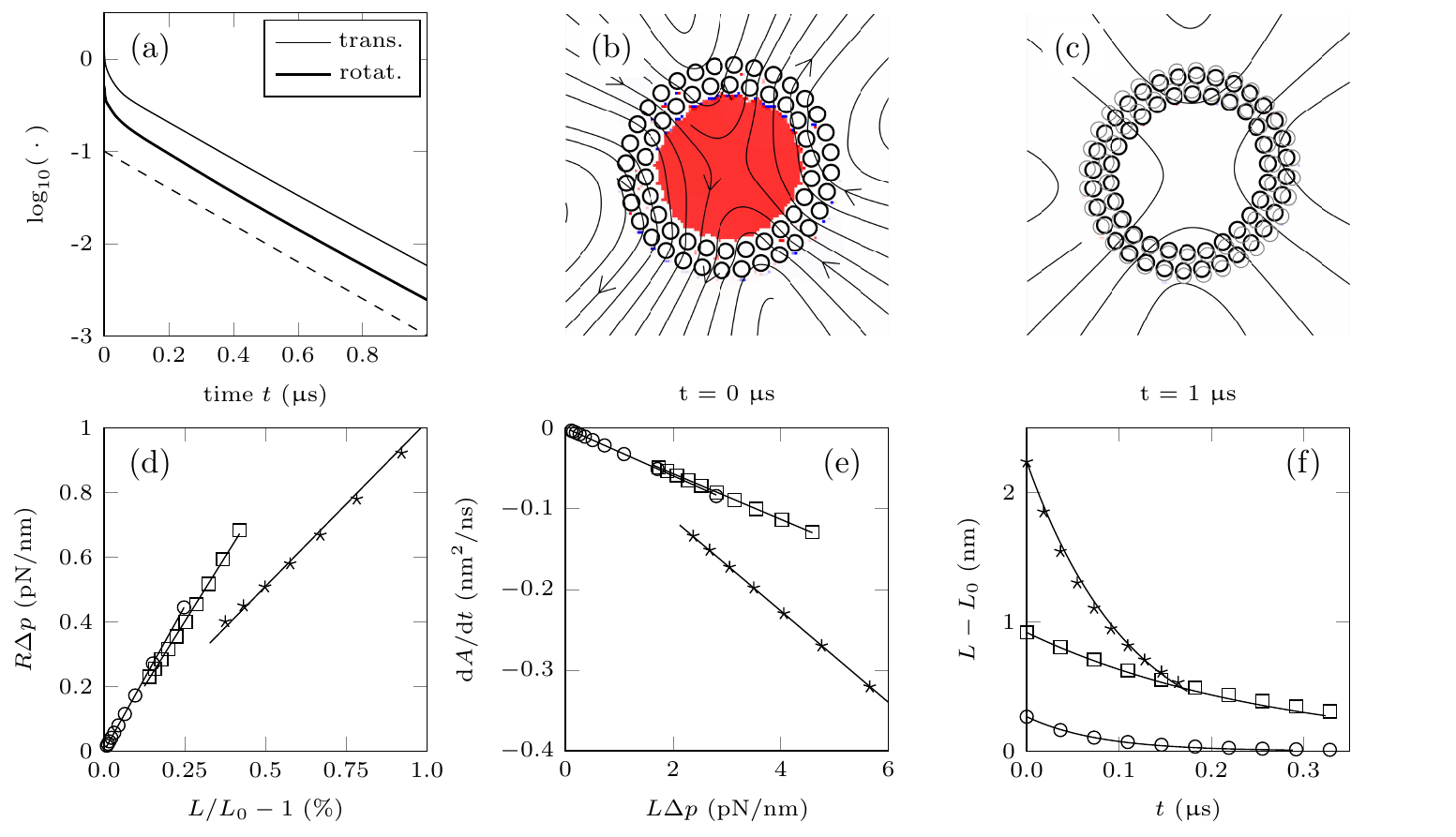}
\caption{\label{figure2} In JP vesicles, tension and pressure relax as
  fluid is expelled through the particle interstices. In panel (a), the
  thin curve plots the mean of the norm of the particle translational
  velocities. The thick curve plots the mean of the absolute value of
  the particle angular velocities. At $t = 0$ $\upmu$s, the vesicle has
  a spatially constant internal pressure 0.05~pN nm$^{-2}$ (red color,
  panel (b)). In panel (c), pressure has dissipated and the vesicle has
  moved from its initial configuration (gray circles), to the
  equilibrium configuration (black circles). The data in panels (d) and
  (e) support the linear relationship expressed by
  \eqref{eq:LaplacePressure} and \eqref{eq:perm}. Circles are for a
  vesicle with 17 nm initial radius; square are for a 34 nm initial
  radius. Asterisks are for a 18 nm initial radius but with smaller,
  1.25 nm diameter particles. In panel (f), the theoretical time course
  \eqref{eq:perm_solution} overlaps the data using constants $K_A$ and
  $P$ derived from panels (d) and (e). }
\end{figure}
Finally, we show that the JP vesicle behaves as a permeable membrane and the inextensibility comes
about due to a large stretching modulus.
Figures~\ref{figure2}b shows an initial, non-equilibrium JP
vesicle suspended in a quiescent flow $\uu_{\infty} = 0$. The color map
plots the pressure $p$. The red color in Figures~\ref{figure2}b shows a
spatially constant, positive internal pressure (0.05~pN nm$^{-2}$) and the white
shows a spatially constant zero, external pressure. There is some fluid flow and
the pressure vanishes as the configuration tends toward the equilibrium
state (Figures~\ref{figure2}c).

What could be the source of this drop in pressure? In membrane
continuum  mechanics, small changes in surface area give rise to a membrane tension
$\tilde \gamma = K_A(A/A_0 - 1)$ where $A$ is the membrane surface area
and $A_0$ is the reference surface area. The area modulus $K_A$ of
bilayers is about $240$ pN nm$^{-1}$ (\cite{NaTr00}). In the
two-dimensional vesicles, the tension becomes
\begin{align}
\label{eq:stretch}
\tilde \gamma = K_A\left(\frac{L}{L_0} - 1 \right),
\end{align}
where $L$ and $L_0$ are the vesicle arc length and resting length,
respectively. Moreover, stretched, circular vesicles has a Laplace pressure 
\begin{align}
\label{eq:LaplacePressure}
\Delta p = \frac{\tilde \gamma}{R},
\end{align}
where $\Delta p$ is the difference in internal pressure to the pressure
at infinity and $R^{-1}$ is the total curvature of the circular
cylinder.

Figure~\ref{figure2}c plots the data for the pressure jump $\Delta p$
between the particle centre and the far-field (see
Figures~\ref{figure2}c and~\ref{figure2}d). We use $R = L/(2\pi)$ for
the vesicle radius, and the horizontal axis is the relative stretch.
The circles are data for a vesicle with radius 17 nm. The linear fit
(solid lines) shows that $R\Delta P$ is proportional to the relative
stretching $L/L_0 - 1$. The squares are for a vesicle with twice the
radius 34 nm, and the data overlap supports that the proportionality
constant $K_A$ is a stretching modulus that is independent of vesicle
size. The data give $K_A = 170 \pm 9$ pN nm$^{-1}$ which is in good
agreement with the experimentally obtained area moduli of lipid bilayer. 
The reason the tank-treading vesicle appears inextensible ($L$ is more or less constant in Figure~\ref{figure4}c) is because the modulus $K_A$ is large.
To evaluate how changes to particle size lead to different
physical properties, the asterisk symbols are for an 18 nm radius
vesicle consisting of particles with diameter 1.25 nm, instead of the
usual 2.5 nm. These data give a smaller stretching modulus of $K_A =
102$~pN~nm$^{-1}$. 

The particles in our setup do not abut but rather have small gaps due to
repulsive forces. The gaps allow for some fluid flux across the JP
bilayer, and in membrane mechanics aqueous flux is quantified by the equation
\begin{align}
  \label{eq:perm} 
  \frac{dA}{dt} = -P L \Delta p,
\end{align}
where $P$ is a hydraulic permeability constant (\cite{chabanon2017,
qua-gan-you2021}). Figure~\ref{figure2}e shows that the data for
$L\Delta p$ and $dA/dt$ obey the linear relationship expressed by
\eqref{eq:perm}. The slopes of the linear fits give the hydraulic
permeabilities $P = 0.0296$ nm$^3$ ns$^{-1}$ pN$^{-1}$ and $P = 0.0283$
nm$^3$ ns$^{-1}$ pN$^{-1}$ for the 17 nm and 34 nm radius cases,
respectively. Like the area modulus, the data give a permeability 
that is independent of vesicle size.

\begin{figure}
\begin{center}
\includegraphics[width=\textwidth]{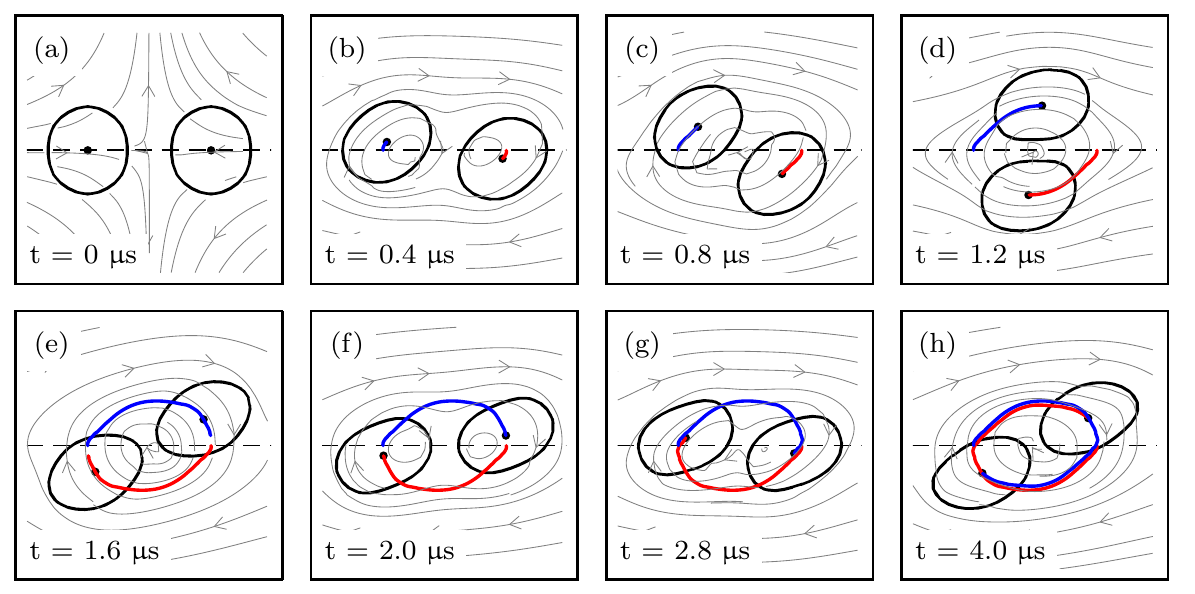}
\end{center} 
  \caption{\label{figure9} Two JP vesicles suspended in a shear flow
  interact by orbiting about the origin. The moving paths of the two
  centroids are plotted in blue and red. The centroids are initially
  located on the $x$-axis. The fluid velocity streamlines appear in
  grey.  The dimensionless shear rate is $\chi=0.005$.}
\end{figure}

The hydraulic permeability we calculate, however, is not in agreement
with experimentally derived values and is larger by a few orders of
magnitude. We suspect this discrepancy is due to inter-particle distance
of the JP being large compared to the inter-lipid spacing in real
bilayers. To test this, we calculate the permeability for the particles
with diameter 1.25 nm. Due to their smaller size but fixed repulsion
strength, these particles have a larger inter-particle spacing resulting
in an increase in permeability, $P = 0.0566$~nm$^3$~ns$^{-1}$~pN$^{-1}$
(Figure~\ref{figure2}e, asterisk symbols).

Since the vesicles in Figure~\ref{figure2}a are nearly circular, we can
combine~\eqref{eq:stretch},~\eqref{eq:LaplacePressure},
and~\eqref{eq:perm}, to derive
\begin{align}
\label{eq:perm_solution}
L_0(L-L(0)) + L_0^2 \ln\left|\frac{L-L_0}{L(0)-L_0}\right| = -4\pi^2 P K_A t.
\end{align}
All in all, the theoretical time courses for \eqref{eq:perm_solution} are in good
quantitative agreement with the JP data (Figure~\ref{figure2}f).

\subsection{Two Vesicles in a Linear Flow}

\begin{figure}
  \centering
\includegraphics[width=\textwidth]{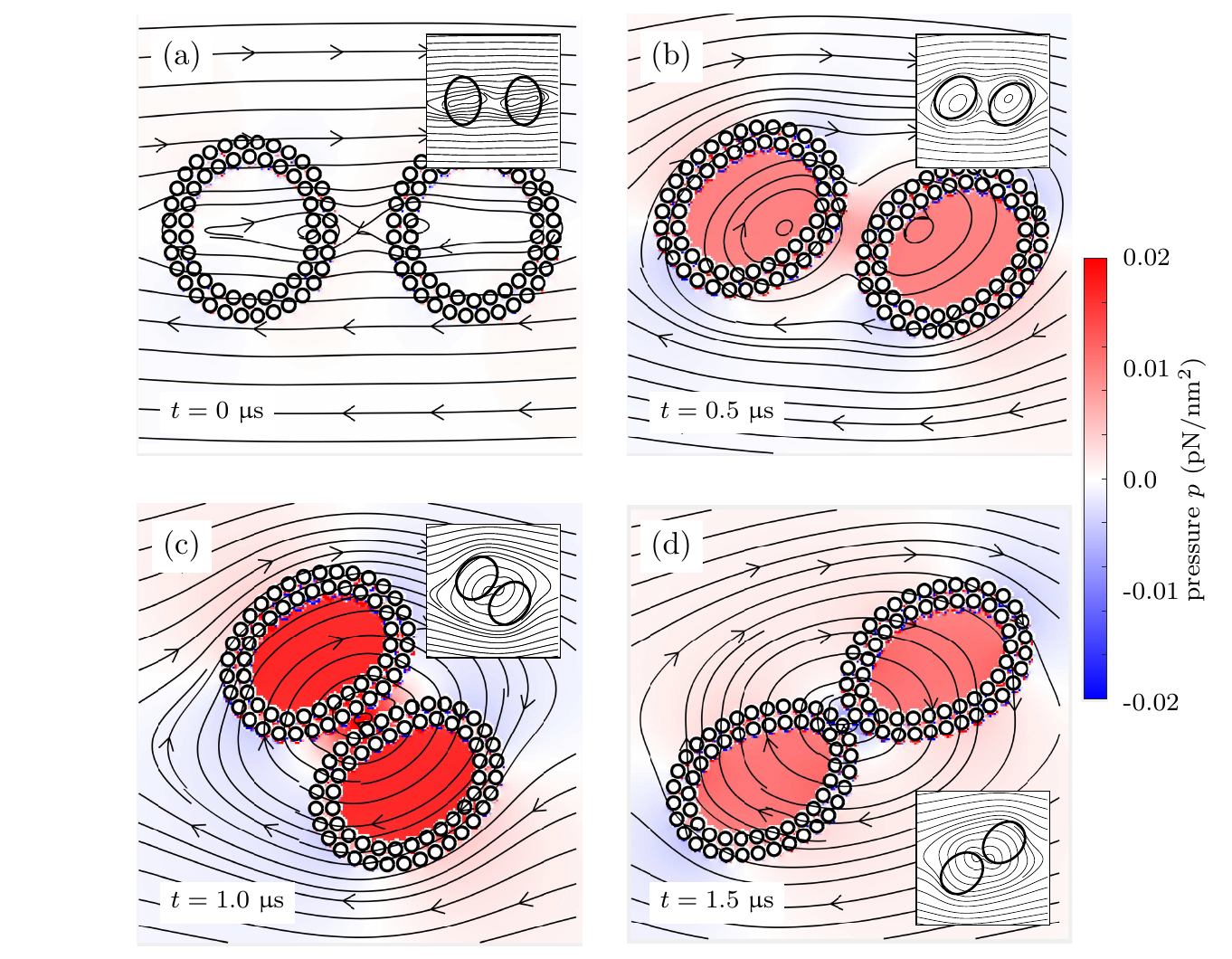}  
  \caption{\label{figure10} 
  The figure is for the same simulation as in Figure~\ref{figure9}, 
  but also plots the fluid pressure $p$. The color bar is identical for all panels. 
  Internal pressure is initially zero, but then grows as the vesicles come close to contact  (panel (c)).
  The insets are generated from simulations of a continuum model of the vesicles (\cite{qua-vee-you2019}). 
  The streamlines of the two models are in agreement.}
\end{figure}

\subsubsection{Shear Flow}
Finally, we can study vesicle-vesicle interactions in background flows. 
Figure~\ref{figure9} shows the simulation of two JP vesicles suspended
in a shear flow with shear rate $\chi=0.005$. We duplicate the
pre-relaxed $58$-body JP vesicle from previous sections and construct the
initial configuration shown in Figure~\ref{figure9}a. The two centroids
are at coordinates $(-25,0)$ and $(25,0)$ in nm. In all panels, the blue and red
curves show the trajectory of the two JP vesicle centroids. They have
nearly completed a full period by $t=4$ $\upmu$s. 

We show snapshots of the fluid pressure in Figure~\ref{figure10}. Since the
initial JP vesicles are pre-relaxed, there is initially no pressure jump
between the internal and external fluids (panel (a)). Panels (b)--(d)
show the configurations when $t = \{0.5,1,1.5\}$ $\upmu$s, and the
streamlines are plotted in the background for all panels. 
We include the numerical results from a continuum model in all insets and these
comparisons give a qualitative agreement between two models.
We also observe an adhesive
effect between the two JP vesicles that is set up by the hydrophobic
attraction (Figure~\ref{figure10}c). Similar dynamics have been observed between a pair of adhering vesicles in a shear flow
(\cite{qua-vee-you2019, abb-far-ezz-ben-mis2021}). This adhesive
behaviour is absent when two JP vesicles are well-separated.

\begin{figure}
\begin{center}
\includegraphics[width=0.97\textwidth]{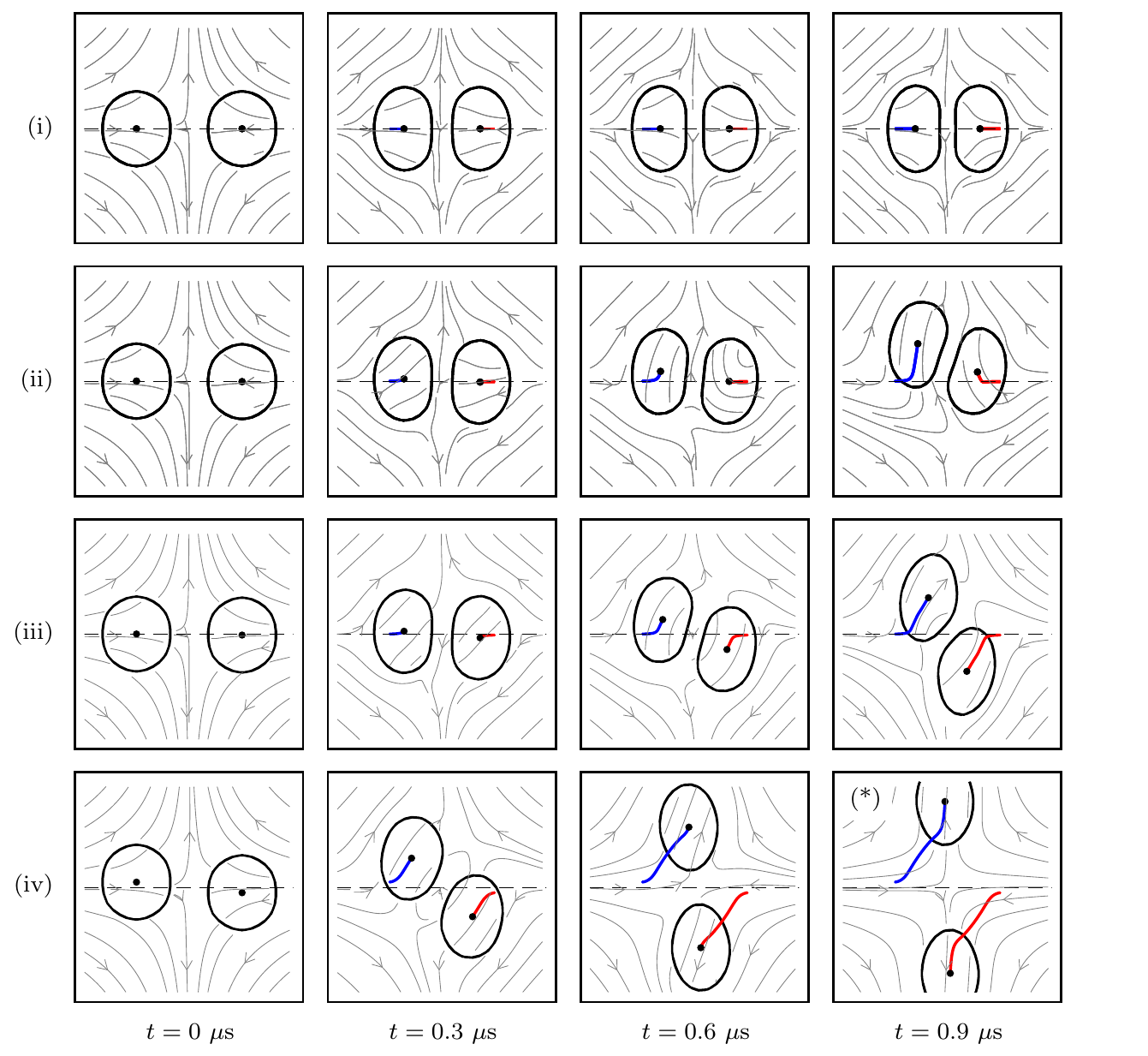}
\end{center} 
  \caption{\label{figure11} The initial placement of two JP vesicles
  suspended in an extensional flow affects the long-time dynamics. The
  moving paths of the two centroids are plotted in blue in red. Each row
  corresponds to a different initial placement of the JP vesicles. (i)
  Both JP vesicle centroids are on the $x$-axis. (ii) The centroid of
  the left JP vesicle is $0.25$~nm above the $x$-axis and the right JP
  vesicle is on the $x$-axis. (iii) The centroid of the left JP vesicle
  is $0.25$~nm above the $x$-axis and the right JP vesicle is $0.25$~nm
  below the $x$-axis. (iv) The centroid of the left JP vesicle is
  $1.25$~nm above the $x$-axis and the centroid of the right JP vesicle
  is $1.25$~nm below the $x$-axis. The streamlines appear in grey and
  the flow rate is $\dot \gamma =0.005$~ns$^{-1}$ in all cases.}
\end{figure}

\begin{figure}
  \centering
  \includegraphics[width=\textwidth]{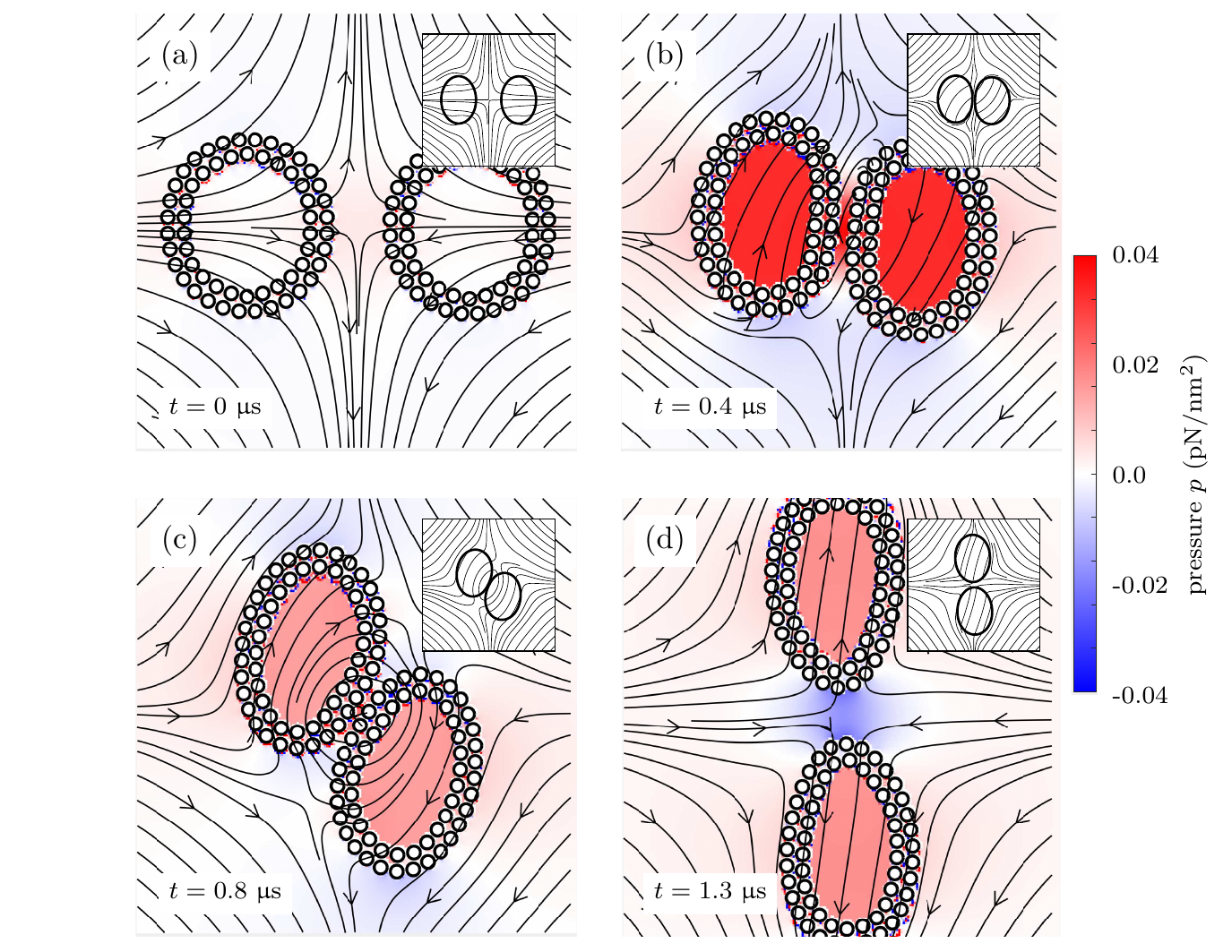}
  \caption{\label{figure12} The figure is for the simulation case (iii) of Figure~\ref{figure11},
  and additionally plots the pressure $p$.  The color bar is shared by all panels. 
  The insets of all panels are generated from simulations of a continuum
  model of the vesicles and the streamlines agree between the two models.}
\end{figure}

\subsubsection{Extensional Flow}

Using a similar setup to the shear flow case, we suspend the same two,
pre-relaxed, $58$-body JP vesicles and consider their dynamics in an
extensional flow 
\begin{equation}
\uu_{\infty}(\xx) = \dot \gamma ((\mathbf{e}_x \cdot \xx)\mathbf{e}_x - (\mathbf{e}_y \cdot \xx)\mathbf{e}_y)
\end{equation}
with extensional rate $\dot \gamma =0.005$ ns$^{-1}$. This extensional
flow is stretching in the $y$-direction and squeezing in the
$x$-direction. Figure~\ref{figure11} shows how the initial placement of
the JP vesicles affects the dynamics. When two centroids are both placed
symmetrically on the $x$-axis (case (i)), the JP vesicles come into
contact and reach a steady equilibrium. If one centroid is placed above
the $x$-axis (case (ii)), the two JP vesicle move together and then
upward. The migration of the right JP vesicle is a consequence of the
adhesive effect caused by the hydrophobic interactions. When the two
centroids start on opposite sides of the $x$-axis (case (iii)), they
eventually diverge from one another along the $\pm y$-directions.
Finally, when the two centroids start on opposite sides of the $x$-axis,
but with a greater displacement (case (iv)), the JP vesicles move much
faster along the $\pm y$-directions.

Figure~\ref{figure12} shows numerical results when the centroids of the
two Janus particles are placed at $(-10,-0.1)$~nm and $(10,0.1)$~nm
(case (iii) from Figure~\ref{figure11}). With this setup, the two JP
vesicles eventually separate along the $\pm y$-directions and we compare
the results against a continuum model as shown in all insets. Panels
(b)--(d) show the transient behaviour of the JP vesicles and the
continuum vesicles under an extensional flow. In both the coarse-grained
model and the continuum model, the vesicles initially converge towards
one other and then diverge along the $y$-axis. The behaviour of the
streamlines in both cases are similar. The pressure is initially largest
in the gap formed by two JP vesicles and decreases during the
separation. The short-range repulsion plays an important role to avoid
particle collisions.


\newpage
\section{\label{conclusion}Conclusion}

\cite{Fu20} developed a mathematical model to quantify the macroscopic
assembly and mechanics of a JP vesicle in a viscous solvent. The
interactions between JP are formulated as a second-kind integral
equation, which is coupled to the Stokes equations for the surrounding
incompressible fluid at the zero-Reynolds-number limit. Numerical
simulations of a JP suspensions revealed self-assembly of JP into
micelles and bilayers, providing an alternative means for computing
mechanical moduli, which often requires the knowledge of an equation of
state from experiments on a colloidal
membrane~(\cite{Balchunas2019_SM}).
%
%
Results in this work show great potential to study Janus colloids~(\cite{Bradley2017,Mallory2017}) and the morphology of colloid
surfactants~(\cite{Bradley2016}). 
%
For example, with the flexibility of the model, we can specify the
boundary condition on JP surfaces based on the chemicals used in
experiments.


In the present study, we used this integral formulation and numerical
algorithm to simulate the hydrodynamics of JP vesicles in background
flows. Under a linear shear flow, we found a JP vesicle to exhibit
elongation and tank-treading dynamics observed for a lipid bilayer GUV.
%
%
The results showed that the reduced area $A^*$ decreases with shear rate
but that the total length of a JP vesicle is conserved. The decay rate
of the reduced area was independent of the shear rate values between
$0.003$~ns$^{-1}$ and~$0.005$~ns$^{-1}$. Moreover, the proposed model
describes membrane rupture in high shear rates. Therefore, our method
can be applied to vesicles undergoing topological changes which is
difficult to simulate when using a continuum model that represents
vesicles as closed and continuous curves.
%
%

%
%
We estimated the inter-monolayer friction $b$, membrane permeability
constant $P$, and the membrane stretching modulus $K_A$. The
inter-monolayer friction coefficient was determined by calculating the
tangential shear force and slip velocities with respect to the bilayer
mid-plane. The range of friction coefficients agree with values reported
by~\cite{denOtter2007} in their MD study. The coarse-graining level of
the JP vesicle has a larger length scale than molecular dynamics
simulations, and in the future convergence studies we will investigate
how physical properties like the friction coefficient and membrane
permeability depend on the particle shape and size.

We also simulate the spatial migration of a JP vesicle in a parabolic
shear flow. 
%
Replicating the hydrodynamics of a GUV in a Poiseuille
flow~(\cite{Kaoui09, dan-vla-mis2009, cou-kao-pod-mis2008}), the JP
vesicle moves toward the centre of the shear flow. While the initial
reduced area of the JP vesicle is $A^* \approx 1$, the equilibrium
reduced area is $A^*=0.9$. For the parameters we used in the simulation,
the JP vesicle takes on an asymmetric, ``slipper" shape as it settles
above the centre of flow and exhibits tank-treading motion.
An interesting result contrast with continuum results is that the JP
vesicle oscillates at a height slightly above the centre of the
Poiseuille flow. 


We further simulated the hydrodynamics of two JP vesicles, and drew
comparisons with simulation results of two vesicles described by the
Helfrich continuum model. A comparison of the vesicle shapes and the
streamlines demonstrate remarkable similarities. The two overlapping
trajectories of the JP vesicles' centroids in a shear flow evolve as
expected when the two centroids are initialized on the same horizontal
level. We also observe a rotating behaviour that is observed for models
involving vesicle adhesion~(\cite{qua-vee-you2019}). The hydrophobic
attraction led to this adhesive effect when two JP bilayers are
sufficiently close. We also performed several simulations of a pair of
JP vesicles suspended in an extensional flow. By varying the initial
vertical displacement of the vesicles' centroids, we can control for
divergent trajectories and obtain similar results to the continuum
model.

In \S~\ref{subsec:calculating_force}, we derived an alternative integral
form for calculating the force and torque to avoid the singular integral
evaluation. These alternative integrals allow us to accurately resolve
trajectories over long times without having to rely on computationally
expensive quadratures.

Our future goals include extending the current framework to a
three-dimensional JP vesicle system. This will require additional
algorithmic implementation including a fast summation method such as the
fast multipole method. Another research direction is to include the
fluctuating hydrodynamics for Brownian suspensions~(\cite{Bao2018}),
and this is critical to understand membrane diffusion. Finally, a more
physical boundary conditions for the HAP model will allow us to draw
comparisons between computational and laboratory experiments.

\begin{acknowledgments}
  {\bf Acknowledgments:} B.Q.~acknowledges support from NSF (Grant
  No.~DMS 2012560) and from the Simons Foundation, Mathematics and
  Physical Sciences-Collaboration Grants for Mathematicians (Award
  No.~527139). Y.-N. Y.~acknowledges support from NSF (Grant No.~DMS 1614863 and 1951600) and
  Flatiron Institute, part of Simons Foundation.
\end{acknowledgments}

%

\bibliographystyle{jfm}

\providecommand{\noopsort}[1]{}\providecommand{\singleletter}[1]{#1}%

\end{document}